\documentclass[letterpaper, 10 pt, conference]{ieeeconf}

\usepackage[noadjust]{cite}
   
\pdfoutput=1

\IEEEoverridecommandlockouts                               
\let\labelindent\relax

\usepackage{amsmath}
\usepackage{amssymb}
\usepackage{amsthm}
\usepackage{mathrsfs}
\usepackage{mathtools}
\usepackage{enumitem}
\usepackage{xfrac}
\usepackage{float,xcolor,color}
\usepackage[ruled,linesnumbered,vlined]{algorithm2e}
\let\oldnl\nl
\newcommand{\nonl}{\renewcommand{\nl}{\let\nl\oldnl}}
\usepackage{soul}
\usepackage{epstopdf}

\newcommand{\emnote}[1]{{\color{blue}[Erik -- #1]}}

\newcommand{\kznote}[1]{{\color{magenta}[Kaiqing -- #1]}}

\newcommand{\EE}{\mathbb{E}}

\newcommand{\Ns}{\mathcal{N}}
\newcommand{\Ms}{\mathcal{M}}
\newcommand{\Ts}{\mathcal{T}}

\newcommand{\Zs}{\mathcal{Z}}
\newcommand{\RR}{\mathbb{R}}

\newtheorem{theorem}{Theorem}
\newtheorem{lemma}{Lemma}
\newtheorem{assumption}{Assumption}
\newtheorem{definition}{Definition}

\DeclareMathOperator{\argmin}{argmin}

\makeatletter
\def\old@comma{,}
\catcode`\,=13
\def,{%
  \ifmmode%
    \old@comma\discretionary{}{}{}%
  \else%
    \old@comma%
  \fi%
}
\makeatother

\newtheorem{example}{Example}
\title{\LARGE \bf 
Approximate Equilibrium Computation  for Discrete-Time \\
Linear-Quadratic Mean-Field Games
} 

\author{Muhammad~Aneeq~uz~Zaman, Kaiqing~Zhang, Erik~Miehling, and Tamer~Ba{\c s}ar
\thanks{The authors are affiliated with the Coordinated Science Laboratory, University of Illinois at Urbana--Champaign Urbana, IL 61801.}
\thanks{Research supported  in part by AFOSR (FA9550-19-1-0353),   in part by ARL (W911NF-17-2-0196),and in part by ARO (W911NF-16-1-0485).}}

\begin{document}
\maketitle
\thispagestyle{empty}
\pagestyle{empty}

\begin{abstract}
While the topic of mean-field games (MFGs) has a relatively long history, heretofore there has been limited work concerning algorithms for the computation of equilibrium control policies. In this paper, we develop a computable policy iteration algorithm for approximating the mean-field equilibrium in linear-quadratic MFGs with discounted cost. Given the mean-field, each agent faces a linear-quadratic tracking problem, the solution of which involves a dynamical system evolving in retrograde time. This makes the development of forward-in-time algorithm updates challenging. By identifying a structural property of the mean-field update operator, namely that it preserves sequences of a particular form, we develop a forward-in-time equilibrium computation algorithm. Bounds that quantify the accuracy of the computed mean-field equilibrium as a function of the algorithm's stopping condition are provided. The optimality of the computed equilibrium is validated numerically. In contrast to the most recent/concurrent results, our algorithm appears to be the first to study \emph{infinite}-horizon MFGs with \emph{non-stationary} mean-field equilibria, though with focus on the linear quadratic setting.
\end{abstract}

\section{Introduction}
\label{sec:intro}
Recent years have witnessed the tremendous progress of operation, control, and learning  in multi-agent systems   \cite{shoham2008multiagent,wooldridge2009introduction,dimarogonas2011distributed,zhang2018fully,zhang2018finite}, where multiple agents strategically  interact with each other in a common environment, to  optimize either a common or individual long-term return. {Despite the substantial interest}, most {existing} \emph{algorithms} for multi-agent  systems suffer from \emph{scalability issue{s}}, {due to} their complexity {increasing} exponentially with the number of agents involved. This  issue has precluded the application of many algorithms {to systems with even a moderate number of agents, let alone} to real-world applications \cite{breban2007mean,couillet2012electrical}. 
 
One way to address {the scalability issue} is {to view the problem in the context of} \emph{mean-field games} (MFGs),   proposed {in} the seminal works {of} \cite{huang2006large,huang2003individual} and, {independently}, \cite{lasry2007mean}. {Under the}  mean-field  {setting}, the interactions among   {the}  agents  {are} approximately represented by the {distribution of} all agents' {states}, {termed the mean-field}, where the influence of each agent on the system is assumed to be infinitesimal in the large population setting. In fact, the more agents are involved, the more accurate the mean-field approximation is, offering an effective tool for addressing the scalability issue. Moreover, following the so-termed \emph{Nash certainty equivalence (NCE) principle} \cite{huang2006large}, the solution to an MFG, referred to as {a} \emph{mean-field equilibrium (MFE)}, can be {determined by} each agent   computing a best-response {control policy to some mean-field that is consistent with the aggregate behavior of all agents. This principle decouples  {the} process of finding the solution {of the game} into  {a} computation{al procedure} of {determining the} best-response {to} a fixed mean-field at  {the} agent {level}, and an update of the mean-field for all agents. In particular,  a  {straightforward} routine for computing the MFE proceeds as follows:  first, each agent  calculate{s} the optimal control, best-responding to some given mean-field, and then, after executing the control, the states are aggregated to update the mean-field.  {This} routine {is referred to} as the \emph{NCE-based approach}, which  {serves as the foundation for} our algorithm. 

Serving as a standard, but significant, benchmark for general MFGs, linear-quadratic MFGs (LQ-MFGs) \cite{huang2007large,bensoussan2016linear,huang2018linear} have been advocated in the literature. In particular, the cost function describing deviations in the state, from the mean-field, as well as the {cost for a given} control {effort} is assumed to be quadratic while the transition dynamics are assumed {to be} linear. Intuitively, the cost  {incentivizes} each agent to \emph{track} the collective behavior of the population, which, for any fixed mean-field, leads to a \emph{linear-quadratic tracking} (LQT) subproblem for each agent. Though simple in form, equilibrium computation in LQ-MFGs (most naturally posed in continuous state-action spaces) inherits most of the challenges from equilibrium computation in general MFGs. While much work has been done in the continuous-time setting \cite{huang2007large,bensoussan2016linear,huang2018linear}, the discrete-time counterpart has received considerably less attention. It appears that, {only the work of \cite{moon2014discrete} (which considered a model with \emph{unreliable communication}  {with an average cost criterion}) has studied a} discrete-time version of the model proposed in \cite{huang2007large}. {The formulation of the discrete-time model of our paper, and the associated equilibrium analysis, are in a setting distinct from \cite{moon2014discrete}, and constitute one of the contributions of the present work.}

There has been an increasing interest in developing (model-free) equilibrium-computation  algorithms for certain MFGs  \cite{subramanian2019reinforcement,guo2019learning,elie2019approximate,fu2019actor}; see \cite[Sec. 4]{zhang2019multi} for more a detailed summary. The closest setting to ours is   in the concurrent while independent work on learning for discrete-time LQ-MFGs   \cite{fu2019actor}. However, given any fixed mean-field,  \cite{fu2019actor} treats each agent's subproblem as a \emph{linear quadratic regulator (LQR) with drift}, which is different from the continuous-time formulation \cite{huang2007large,bensoussan2016linear,huang2018linear}. {This is made possible because they considered   mean-field {trajectories} that are constant  {in} time  {(also referred to as \emph{stationary} mean-fields)}}. This is in contrast to the LQT subproblems found in both the literature \cite{huang2007large,bensoussan2016linear,huang2018linear} and in our formulation. While the former admits a \emph{forward-in-time} optimal control that can be obtained using policy iteration   and standard reinforcement learning (RL) algorithms  \cite{bradtke1993reinforcement,fazel2018global,zhang2019policy}, the latter leads to a \emph{backward-in-time} optimal control problem, which{, in general,} has been recognized to be challenging to solve,  especially in a model-free fashion \cite{kiumarsi2014reinforcement,modares2014linear}. 
Most other RL algorithms for general MFGs are also restricted to the stationary mean-field setting \cite{subramanian2019reinforcement,guo2019learning}, which does not apply to the LQ-MFG problem here. 
Fortunately, by identifying  {a} structural property of  our policy iteration algorithm and employing an NCE-based equilibrium-computation approach, one can develop a computable algorithm that executes forward in time. 

\vspace{5pt}
\noindent\textbf{Contribution.}
Our contribution in this paper is three-fold: ({\bf 1}) We formally introduce the  formulation of discrete-time LQ-MFGs with discounted cost, complementing  the standard continuous-time formulation \cite{huang2003individual,huang2007large}, and the discrete-time   average-cost setting of \cite{moon2014discrete}, together with existence and uniqueness guarantees for the MFE. ({\bf 2}) By identifying  structural results of the NCE-based  policy iteration update, we develop an equilibrium-computation algorithm, with convergence error analysis, that can be implemented   \emph{forward-in-time}. ({\bf 3})  We illustrate the quality of the computed MFE in terms of the algorithm's stopping condition and the number of agents. Our structural results and equilibrium-computation algorithm lay  foundations for developing model-free RL algorithms, as our immediate  future work. 

\vspace{5pt}
\noindent\textbf{Outline.}
The remainder of the paper proceeds as follows. In Section II, we introduce the linear-quadratic mean-field game model. Section III provides a background of relevant results from the literature on mean-field games as well as establishes a characterization of the mean-field equilibrium for our setting. Section IV outlines some properties of the computational process and presents the algorithm. Numerical results are presented in Section V. Concluding remarks and some future directions are presented in Section VI. Proofs of all results have been relegated to the Appendix. 

\section{Linear Quadratic Mean-Field Game Model}
\label{sec:model}

Consider a dynamic game with $N<\infty$ agents playing on an infinite time horizon. For each agent $n \in [N]$, let $z_t^n \in \RR$ represent the current state and $u_t^n \in \RR$ represent the current control. Each agent $n$'s state is assumed to follow linear time-invariant (LTI) dynamics,
\begin{align}
z_{t+1}^n = az_t^n +bu_t^n + w_t^n\label{eq:finitesystem},
\end{align} 
with constants $a\in\RR$, $b\in\RR\setminus\{0\}$, independent and identically distributed initial state $z_0^n$ with mean $\nu_0$ and variance $\sigma_0^2$, and independent identically distributed noise terms, $w_t^n\sim\Ns(0,\sigma_w^2)$, assumed to be independent of $z_0^{n'}$, $w_s^{n'}$ for all $s$ and $t$, and for all $n'\neq n$.

At the beginning of each time step, each agent observes every other agent's state. Thus, assuming perfect recall, the information of agent $n$ at time $t$ is $i_t^n = \big((z_0^1,\ldots,z_0^N),u_0^n,\ldots,(z_{t-1}^1,\ldots,z_{t-1}^N),u_{t-1}^n,(z_t^1,\ldots,z_t^N)\big)$. A control policy for agent $n$ at time $t$, denoted by $\eta_t^n$, maps its current information $i_t^n$ to a control action $u_t^n\in\RR$. The joint control policy is the collection of policies across agents, and is denoted by $\eta_t = (\eta_t^1,\ldots,\eta_t^N)$. The joint control law is the collection of joint control policies across time, denoted by $\eta = (\eta_0,\eta_1,\ldots)$.

The agents are coupled via their expected cost functions. The expected cost for agent $n$ under joint policy $\eta$ and the initial state distribution, denoted by $J^n(\eta)$, is defined as, 
\begin{align} \label{eq:costfcni}
	\nonumber&J^n(\eta) :=\\
	&\hspace{0.5em}\sum_{t=0}^{T}\gamma^t \EE_{\eta}
	\Bigg[  c_z\big(z_t^n - \frac{1}{N-1}\sum_{n'\neq n}z_t^{n'}\big)^2+c_u(u_t^n)^2 \Bigg],
\end{align}
where $\gamma\in[0,1)$ is the discount factor and $c_z, c_u > 0$ are cost weights for the state and control, respectively. The expectation is taken with respect to the randomness of all agents' state trajectories induced by the joint control law $\eta$ and the initial state distribution.

In the finite-agent system described above, each agent is assumed to fully 
observe all other agents' states. As $N$ grows, determining a policy that is a 
best-response to all other agents' policies becomes computationally intractable, 
precluding computation of a Nash equilibrium \cite{cardaliaguet2018mean}. 
Fortunately, since the coupling between agents manifests itself as an average of all 
agent's states, one can approximate the finite agent game by an infinite 
population game in which a \emph{generic agent} interacts with the mass behavior of all agents. 
The empirical average of all agents' states becomes the mean state process 
(\emph{i.e.}, the \emph{mean-field}), decoupling the agents and yielding a 
stochastic control problem. The infinite population game is termed a \emph{mean-field game} \cite{huang2006large}. In this paper, we focus on \emph{linear-quadratic} MFGs  in which the generic agents' dynamics are linear and its costs 
are quadratic. 

The state process of the generic agent is identical to  \eqref{eq:finitesystem}, that is, 
\begin{align} \label{eq:dynamics}
z_{t+1} = az_t +bu_t + \omega_t,
\end{align}
where $z_0$ is distributed  with mean $\nu_0$ and variance $\sigma_0^2$, and $\omega_t$ is an i.i.d. noise process generated according to the distribution $\Ns(0,\sigma_w^2)$, assumed to be independent of the mean-field and the agent's state. 

The generic agent's control policy at time $t$, denoted by $\mu_t$, translates the available information at time $t$, denoted by $i_t = (z_0,u_0,\ldots,z_{t-1},u_{t-1},z_t)$, to a control action $u_t\in\RR$. The collection of control policies across time is referred to as a control law and is denoted by $\mu = (\mu_0,\mu_1,\ldots)\in\Ms$ where $\Ms$ is the space of admissible control laws. The generic agent's expected cost  under control law $\mu$ is defined as,
\begin{align} \label{eq:costfcn}
J(\mu,\bar z) = \sum_{t=0}^{\infty} \gamma^t \EE_{\mu}\left[ c_z( z_t-\bar{z}_t )^2 + c_u u_t^2\right],
\end{align}
where $\bar z_t = \EE[z_t]$ represents the mean-field at time $t$. The mean-field {trajectory} $\bar{z}:= (\bar{z}_0, \bar{z}_1,\ldots)$ is assumed to belong to the space of bounded sequences, that is, $\bar{z}\in\Zs$ where $\Zs:= 
\ell^\infty=\{ x = (x_0, x_1, \ldots) \mid \sup_{t\ge0}|x_t| < \infty \}$.

To define a mean-field equilibrium, first define the operator $\Lambda:\Ms\to \Zs$ as a mapping from the space of admissible control laws $\Ms$ to the space of mean-field trajectories $\Zs$. Due to the information structure of the problem, the policy at any time only depends upon the current state \cite{moon2014discrete}. It is defined as follows: given $\mu \in \Ms$, the mean-field $\bar{z} := \Lambda(\mu)$ is constructed recursively as
\begin{align}\label{eq:def_Lambda}
\bar{z}_{t+1} := A \bar{z}_{t} +B \mu_{t}(\bar{z}_{t}),\quad \bar z_0 = \nu_0.
\end{align}
Similarly, define an operator $\Phi : 
\Zs \to \Ms$ as a mapping from a mean-field trajectory to its 
optimal control law,
\begin{align}
\Phi(\bar{z}) := \argmin_{\mu} J(\mu, \bar{z}).
\end{align}
A mean-field equilibrium can now be defined.
\begin{definition} [\cite{saldi2018markov}] \label{def:mfe}
	The tuple $(\mu^*, \bar{z}^*) \in \Ms \times \Zs$ is an MFE if $\mu^*= 
	\Phi(\bar{z}^*)$ and $\bar z^*=\Lambda(\mu^*)$. %
\end{definition}

The power of mean-field analysis is the fact that the equilibrium policies obtained in the infinite-population game are good approximations to the 
equilibrium policies in the finite-population game 
\cite{huang2006large,huang2003individual,lasry2007mean}. The focus of the 
current paper is on approximate equilibrium computation and, while we do not 
derive explicit bounds for finite $N$, we offer empirical results in Section  
\ref{sec:Sims} illustrating the effectiveness of the mean-field approximation.

\section{Background: MFE Characterization} 
\label{sec:MFE}

This section establishes some properties of mean-field equilibria. The results are complementary to those of \cite{huang2003stochastic}, \cite{huang2006large}, and \cite{moon2014discrete}. Note that while \cite{moon2014discrete} constructs a discrete-time analogue of \cite{huang2006large}, the model of \cite{moon2014discrete} considers an average-cost criterion, whereas here we consider a discounted-cost criterion, as in \cite{saldi2018markov}.

Recall that in the limiting case, as $N\to\infty$, the problem becomes a 
constrained stochastic optimal control problem. In particular, as described by \eqref{eq:costfcn}, a generic agent aims to find a control law $\mu$ that 
tracks a given reference signal (the mean-field trajectory). This control law, 
hereafter referred to as the \emph{cost-minimizing control}, is characterized in closed-form by the following lemma.

\begin{lemma}
\label{lem:LQT}
Given a mean-field trajectory, $\bar z = (\bar z_0,\bar z_1,\ldots) \in \Zs$, 
the control law that minimizes \eqref{eq:costfcn}, termed the     
\emph{cost-minimizing control}, denoted by $\Phi(\bar z) = (\mu_0(z_0;\bar z),\mu_1(z_1;\bar z),\ldots)$, is given for each $t$ by,\footnote{The cost-minimizing control policy $\mu_t$ (from the cost-minimizing control $\mu$) is denoted by $\mu_t(\cdot;\bar{z})$ to illustrate that it is parameterized by the mean-field trajectory $\bar{z}$.}
\begin{align} \label{eq:u_t}
u_t = \mu_t( z_t;\bar z) := g_p(apz_t + \lambda_{t+1}(\bar z)),
\end{align}
where $g_p := -\gamma b /(c_u + \gamma b^2 p)$, $p$ is the unique positive solution to the discrete-time algebraic Riccati equation (DARE),
\begin{align} \label{eq:s}
p^2 + \big( [(1-\gamma a^2) c_u/(\gamma b^2)] -c_z\big)p - c_zc_u/(\gamma b^2) = 0,
\end{align}
that is
\begin{align}\label{eq:closedform}
p= ( -\alpha + \sqrt{\alpha^2 + 4 \beta}) / 2,
\end{align}
where $\alpha := {{c_u (1-\gamma a^2)}\over {\gamma b^2}} - c_z$, $\beta := {{c_z c_u}\over {\gamma b^2}}\,,$ and the sequence $\{\lambda_t\}$, referred to as the \emph{co-state}, is generated backward-in-time by,
\begin{align} \label{eq:lambda_t}
\lambda_t(\bar z) = \gamma h_p \lambda_{t+1}(\bar z) - c_z \bar{z}_t,
\end{align}
where $h_p := a (1 + bpg_p)$.
\end{lemma}
To ensure the well-posedness of the cost-minimizing controller for mean-field $\bar{z} \in \Zs$, the optimal cost must be bounded \cite{moon2014discrete}. This is true given the following assumption.

\begin{assumption}  \label{asm:boundforcontract}
	Given $\gamma,a,b,c_z,c_u$ and $g_p$, $h_p$, where $p$ is the positive solution of \eqref{eq:s}, as given by \eqref{eq:closedform}, the quantity $T_p:=  |h_p| +  |c_zb g_p/(1 - \gamma h_p)|$ satisfies $T_p<1$. 
\end{assumption}
This assumption is analogous to condition (H6.1) of \cite{huang2003stochastic} for continuous-time settings. Lemma \ref{lem:bdd_and_hurwitz} shows that under Assumption \ref{asm:boundforcontract}, both the co-state process and  the optimal cost are bounded.

\begin{lemma} \label{lem:bdd_and_hurwitz}
\begin{enumerate}
\item If $\lambda_0(\bar z)  = -c_z \sum_{s=0}^{\infty} \left( \gamma h_p \right)^{s} \bar{z}_{s}$ then $\lambda=(\lambda_0,\lambda_1,\ldots) \in\ell^{\infty}$. Moreover, with this initial condition, 
\begin{align} \label{eq:costate}
\lambda_t(\bar z) = -c_z\sum_{s=0}^{\infty} \left( \gamma h_p \right)^{s}\bar{z}_{t+s},\,\text{for }t=0,1,\ldots.
\end{align}
\item Under Assumption \ref{asm:boundforcontract}, $J(\Phi(\bar z),\bar z)$ for 
any $\bar{z} \in \Zs$ is bounded. 
\end{enumerate}
\end{lemma}

Substituting the cost-minimizing control, \eqref{eq:u_t}, into the state equation, \eqref{eq:dynamics}, the closed-loop dynamics are
\begin{align*}
z_{t+1} &= a z_t + b g_p (ap z_t + \lambda_{t+1}(\bar z)) + \omega_t\\
&= h_p z_t + b g_p \lambda_{t+1}(\bar z) + \omega_t.
\end{align*}
Taking expectation, the above equation becomes $\bar{z}'_{t+1} = h_p \bar{z}_t + b g_p \lambda_{t+1}(\bar z)$
for $t=0,1,\ldots$, where $\bar{z}'_0 = \nu_0$. Substitution of the co-state process, \eqref{eq:costate},
yields the following as the mean-field dynamics,
\begin{align} \label{eq:MF-update}
\bar{z}'_{t+1} = h_p \bar{z}_t - c_zb g_p \sum_{s=0}^{\infty} \left( \gamma h_p \right)^s \bar{z}_{t+1+s}.
\end{align}
In the same vein as \cite{huang2006large}, the above can be compactly summarized as an update rule, termed the \emph{mean-field update operator}, on the space of (bounded) mean-field trajectories. The update rule, denoted by $\Ts:\Zs \rightarrow \Zs$, is given by,
\begin{align}
\bar{z}' = \Ts(\bar{z}) := \Lambda(\Phi(\bar{z})).\label{eq:Ts}
\end{align}
The operator outputs an updated mean-field trajectory $\bar{z}'$, using \eqref{eq:def_Lambda}, resulting from the cost-minimizing control for a  
mean-field trajectory $\bar z$, given by \eqref{eq:u_t}. The operator is a 
contraction mapping, as shown below.
\begin{lemma}\label{lem:contract}
Under Assumption 1, the mean-field update operator $\Ts$ is a contraction mapping on $\Zs=\ell^\infty$.
\end{lemma}
Furthermore, iterated application of $\Ts$ results in a fixed point which corresponds to an MFE, as expressed below.
\begin{theorem}\label{thm:MFE_T_equivalence}
A mean-field trajectory $\bar{z}^*$ is a fixed point of $\Ts$,
\begin{align}
\bar{z}^* = \mathcal{T}(\bar{z}^*),
\end{align}
if and only if $(\Phi(\bar{z}^*), \bar{z}^*)$ is an MFE. 
\end{theorem}
As a corollary to the above results, there exists a unique MFE,  by the Banach fixed-point theorem \cite{luenberger1997optimization}. 
Moreover, a straightforward approach for computing the equilibrium, \emph{i.e.}, the fixed-point of $\Ts$, is to iterate the operator $\Ts$ until convergence. Indeed, we note that this process is referred to as \emph{policy iteration} in the continuous-time LQ-MFGs setting of \cite{huang2007large}. 
However, the cost-minimizing control given by Lemma \ref{lem:LQT} needs to be calculated backward-in-time, which makes the update of $\Ts$ in \eqref{eq:Ts} not computable. In fact, to develop model-free learning algorithms, forward-in-time computation is necessary.

In what follows, we investigate properties of the mean-field operator that permit the construction of a \emph{computable policy iteration} algorithm that proceeds forward-in-time.

\section{Approximate Computation of the MFE}
\label{sec:PI}

\subsection{Properties of the Mean-Field Update Operator}
A prerequisite for the development of any algorithm is that the representations of all quantities in the algorithm are finite. Satisfying this requirement in our case is complicated by the fact that both the equilibrium mean-field trajectory and the cost-minimizing control are infinite dimensional (see Def. \ref{def:mfe}). To address the challenge, we represent the infinite sequences by finite sets of parameters. 

The parameterization of the mean-field trajectory is inspired by a property of the update operator. To show this property, consider the following class of sequences.
\begin{definition} \label{def:tau_latent_seq}
A sequence $x=(x_0,x_1,\ldots)$ is said to be a \emph{$\tau$-latent LTI  sequence} if $x_{t+1}=rx_t$ for some $r\in\RR$ for all $t=\tau,\tau+1,\ldots$.
\end{definition}
Any $\tau$-latent LTI sequence, for $\tau<\infty$, can be represented by $\tau+2$ parameters, summarized by the pair $(x_{0:\tau},r)$, where $x_{0:\tau} = (x_0,\ldots,x_{\tau})$. This is illustrated in the following example. 
\begin{example} Consider the following sequence $(x_0,x_1,\ldots)$ where $\phi_0,\phi_1$ are arbitrary functions and $s_0,r\in\RR$,
\begin{align*}
(x_0,x_1,x_2,x_3,x_4,\ldots) &= (s_0,\phi_0(x_0),\phi_1(x_1),rx_2,rx_3,\ldots)\\
&=: (x_{0:2},r).
\end{align*}
The sequence obeys linear dynamics starting at $t=2$. As such, the above sequence is referred to as a $2$-latent LTI sequence and is denoted by $(x_{0:2},r)$. 
\end{example}

Our algorithm is based on the observation that, given any stable\footnote{{Namely, $|r|\leq 1$.}} $\tau$-latent LTI sequence with constant $r$, the mean-field update operator outputs a stable $(\tau+1)$-latent LTI sequence with the same constant $r$, as summarized by Lemma \ref{lem:LTI2LTI} below.
\begin{lemma} \label{lem:LTI2LTI}
	If $(x_{0:\tau},r)$ is a $\tau$-latent LTI sequence with constant $r$ satisfying $|r| \leq 1$, then $(x'_{0:\tau+1},r)$, where $x' = \Ts(x)$, is a $(\tau+1)$-latent LTI sequence with constant {$r$}.
\end{lemma}

{By} Lemma \ref{lem:LTI2LTI}, {each application of operator $\Ts$ increases} the dimension of the mean-field trajectory's parameterization. This allows us to construct an iterative algorithm in which, for any finite iteration, all quantities are computable.

\subsection{A {Computable} Policy Iteration Algorithm}

This section presents a policy iteration algorithm for approximately computing the mean-field equilibrium. The algorithm operates over iterations $k=1,2,\ldots$, where variables at the $k^{\text{th}}$ iteration are denoted by superscript $(k)$. 

As mentioned {in the discussion following} Theorem \ref{thm:MFE_T_equivalence}, iterating the mean-field update operator $\Ts$ yields a process that {converges} to the MFE, though not computable due to the backward-in-time calculation of the cost-minimizing control. To address this issue, we propose an iterative algorithm that operates on parameterized sequences. Motivated by the result of Lemma \ref{lem:LTI2LTI}, by initializing the algorithm with a $\tau$-latent sequence, we can ensure that, after any finite number of iterations, the computed sequence is also $\tau$-latent. {Importantly, this structure allows one to describe the mean-field trajectory at any iteration by a finite set of parameters}. {Furthermore, } the $\tau$-latent LTI structure {allows for} the cost-minimizing control {to} be calculated {forward-in-time}. As a consequence, the aforementioned procedure can {be carried out} in a computable way, provided that the iteration number $k$ remains finite. 

More formally, our {(computable)} policy iteration algorithm proceeds as follows. Without loss of generality, we start with a $0$-latent LTI mean-field trajectory $\bar z^{(0)}$ with $\bar z^{(0)}_0=\nu_0$ at iteration $0$. Thus, at any iteration $k$, by Lemma \ref{lem:LTI2LTI}, the mean-field trajectory $\bar z^{(k)}$ is a $k$-latent LTI sequence. Hence, the cost-minimizing control under $\bar z^{(k)}$ can be written in parameterized form\footnote{With some abuse of notation, we have replaced the (infinite) mean-field trajectory with its parameterized form.} as:
\begin{align}
\hspace{-0.5em}u_t^{(k)} = \mu_t(z_t;(\bar z^{(k)}_{0:k},r)) := g_p\big(apz_t - c_zl_{k}(t,\bar z^{(k)},r)\big)\label{equ:LTI_u}
\end{align}
where
\begin{align*}
l_{k}(t,\bar z,r) :=\left\{
\begin{array}{ll}
\frac{(\gamma h_p)^{k-t} r \bar z_{k}}{1-\gamma h_pr} + q_{k}(t,\bar z) & \text{if }t<k \\
\frac{r^{t-k+1}\bar z_{k}}{1-\gamma h_pr^{(k)}} & \text{if }t\ge k
\end{array}
\right.
\end{align*}
and 
$q_{k}(t,\bar z) := \sum_{s=0}^{k-t-1}(\gamma h_p)^s\bar z_{t+1+s}$.

Note that the {control} expressed in \eqref{equ:LTI_u} has a closed-form (without infinite sums) and is indeed calculated forward-in-time. The mean-field trajectory is then updated by the operator $\Ts$, which first execut{es} the {control} in \eqref{equ:LTI_u}, then aggregates  the generated mean-field trajectory by averaging the states over all agents,
\begin{align} \label{eq:mean-field_upd_lat}
\bar{z}^{(k+1)}_{t+1} = a \bar{z}^{(k)}_t + b u^{(k)}_t,\quad 
\end{align}
where $\bar{z}^{(k+1)}_{0}=\nu_0$, $0\le t\le k$. 
This closes the loop and leads to a computable version of iterating the operator $\Ts$. The details of the algorithm are summarized in Algorithm \ref{alg:PI_algo}.

\begin{algorithm} [h]
	\caption{Policy iteration for LQ-MFGs} \label{alg:PI_algo}
		\KwData{$a,b,c_z,c_u,\gamma,\nu_0$, $|r|\le1$, and $\varepsilon_s>0$}
		{\bf Initialize}: Set $\bar{z}^0$ as a $0$-latent LTI mean-field with $\bar{z}^{(0)}_0 = \nu_0$, $k=0$\;
		$p\leftarrow ( -\alpha + \sqrt{\alpha^2 + 4 \beta}) / 2$, where $\alpha = {{c_u (1-\gamma a^2)}\over {\gamma b^2}} - c_z$\, and $\beta = {{c_z c_u}\over {\gamma b^2}}$\\
		\Indp
		\Indm
		$g \leftarrow -\gamma b/(c_u+\gamma b^2p)$\\
		$h \leftarrow a(1+bpg)$\\
		$T \leftarrow |h|+|bgc_z/(1-\gamma h)|$\\
		\While{ $\max_{0 \leq t \leq k} \left\lvert \bar{z}^{(k)}_t - \bar{z}^{(k-1)}_t \right\rvert > \varepsilon_s(1-T)/T$}{		
		$\bar{z}^{(k+1)}_0 \leftarrow \nu_0$\\
		\For{$ m \in \{ 0,1,\ldots, k\}$}{
			$\bar{z}^{(k+1)}_{m+1} \leftarrow a \bar{z}^{(k)}_m + b \mu_m(\bar{z}^{(k)}_m;(\bar z^{(k)}_{0:k},r)) $
			}
		$k \leftarrow k + 1$}
		\Return  Parameter tuple  $(\bar z^{(k)}_{0:k},r)$ that yields the control  $\mu(\cdot;(\bar z^{(k)}_{0:k},r))$ (see \eqref{equ:LTI_u})
\end{algorithm}

Algorithm \ref{alg:PI_algo} generates iterates $\bar z^{(k)}$ that approach the equilibrium mean-field trajectory $\bar{z}^*$. Furthermore, the minimum number of iterations required to reach a given accuracy can be represented in terms of the desired accuracy, the initial approximation error, the contraction coefficient, and the constant of the linear dynamics. The convergence is summarized by the following theorem. 

\begin{theorem} \label{thm:approx_mfe_bound}
Under Assumption \ref{asm:boundforcontract}, given $\varepsilon_s>0$ there exists a $k^*> K(\varepsilon_s) := \left\lceil \left( \log \varepsilon_s - \log\left\lVert\bar{z}^{(0)} - \bar{z}^{*} \right\rVert_{\infty} \right) /\log T_p \right\rceil$ such that $\left\lVert\bar{z}^{(k^*)} - \bar{z}^{*} \right\rVert_{\infty} < \varepsilon_s$, where $T_p$ was introduced in Assumption \ref{asm:boundforcontract}.
\end{theorem}

\section{Numerical Results}
\label{sec:Sims}

In this section we present simulations to demonstrate the performance of  Algorithm \ref{alg:PI_algo} that approximates the equilibrium mean-field of the LQ-MFG. We use a normal distribution with mean and variance $\nu_0 = 20.0$ and ${\sigma_0^2} = 1.0$, respectively, to generate the initial condition of the generic agent $z_0$. The dynamics of the generic agent are defined as in \eqref{eq:dynamics} and the parameters are $ a = 1.1315$ and $b = 0.7752$. The standard deviation of the noise process is $\sigma_{\omega} = 0.03$. The cost function has the form shown in \eqref{eq:costfcn} with values $c_z = 0.0392$ and $c_u = 1.6864$. The positive solution {of the resulting} Riccati equation, given by \eqref{eq:closedform}. is $p = 0.8787$ {with} $g_p=-0.3227, h_p = 0.8828$ and $T_p = 0.9305<1$. The algorithm starts with initial mean field $\bar{z}^{(0)}$, which is a $0$-latent LTI mean-field with parameters $\bar{z}_0 = \nu_0, \hspace{0.2cm} r = 0.6$.

Figure \ref{fig:approx_mf} shows approximations of the mean-field for different values of $\varepsilon_s$. As shown, for decreasing values of $\varepsilon_s$ the approximations approach the equilibrium mean-field. Interestingly, the algorithm reaches a good approximation ($\varepsilon_s = 0.005$) in a small number of iterations ($k=40$). 
\begin{figure}[h!]
  \includegraphics[width=0.455\textwidth]{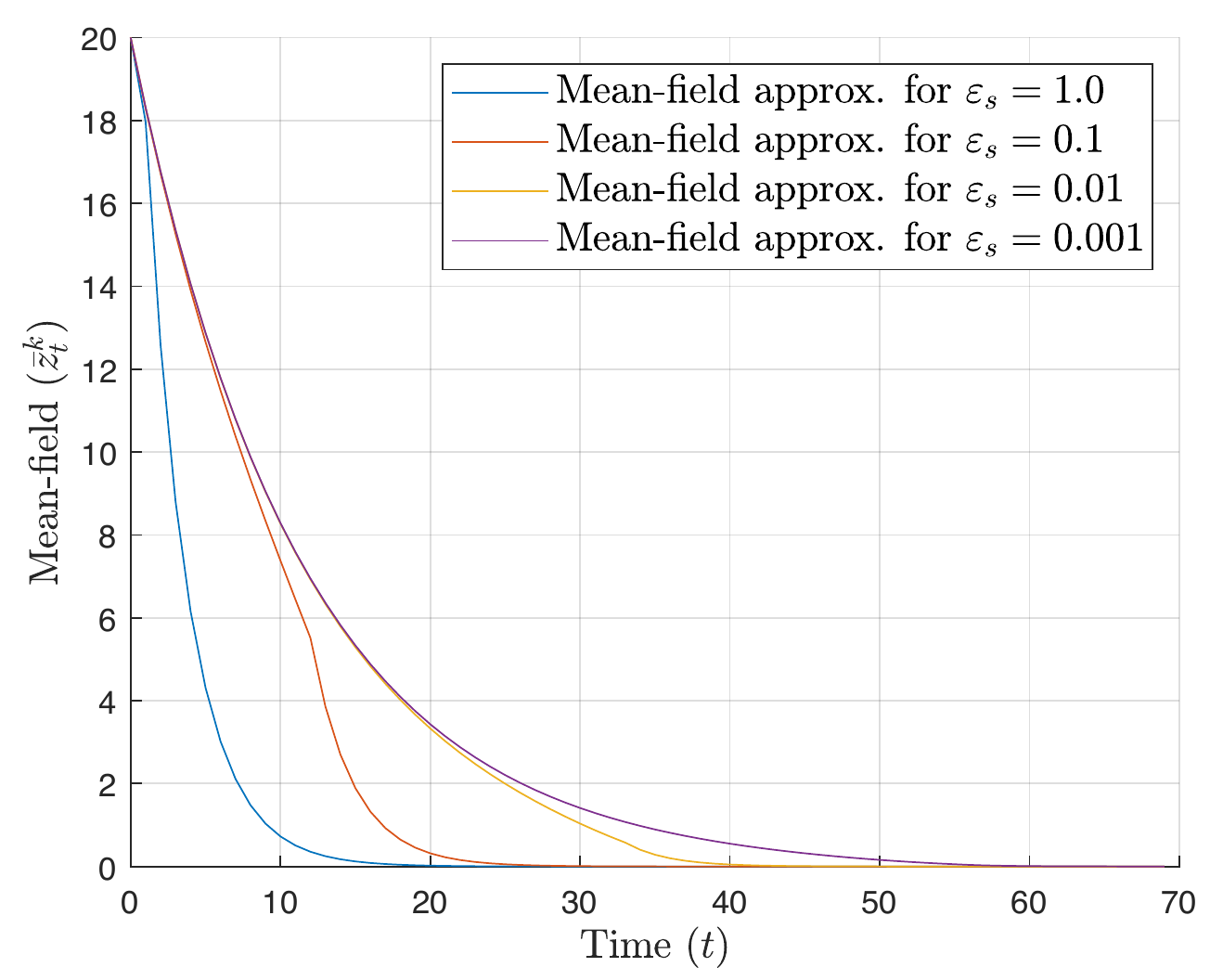}
  \caption{Mean-field approximation for different values of $\varepsilon_s$. {Notice the convergence of the mean-field trajectory as $\varepsilon_s$ decreases.}}
  \label{fig:approx_mf}
\end{figure}
Figure \ref{fig:avg_cost} depicts the average cost per agent for different numbers of agents, $N$ and for different values of $\varepsilon_s$. Each plot in the figure corresponds to a different number of agents $N$. As $N$ increases, the average cost is seen to decrease. This provides evidence that our conjecture, regarding policies obtained from the infinite population case when applied to the finite population case, is correct. The figure also shows that as the approximations become better, there is a decrease in the average cost per agent.

\begin{figure}[h!]
	\includegraphics[width=0.45\textwidth]{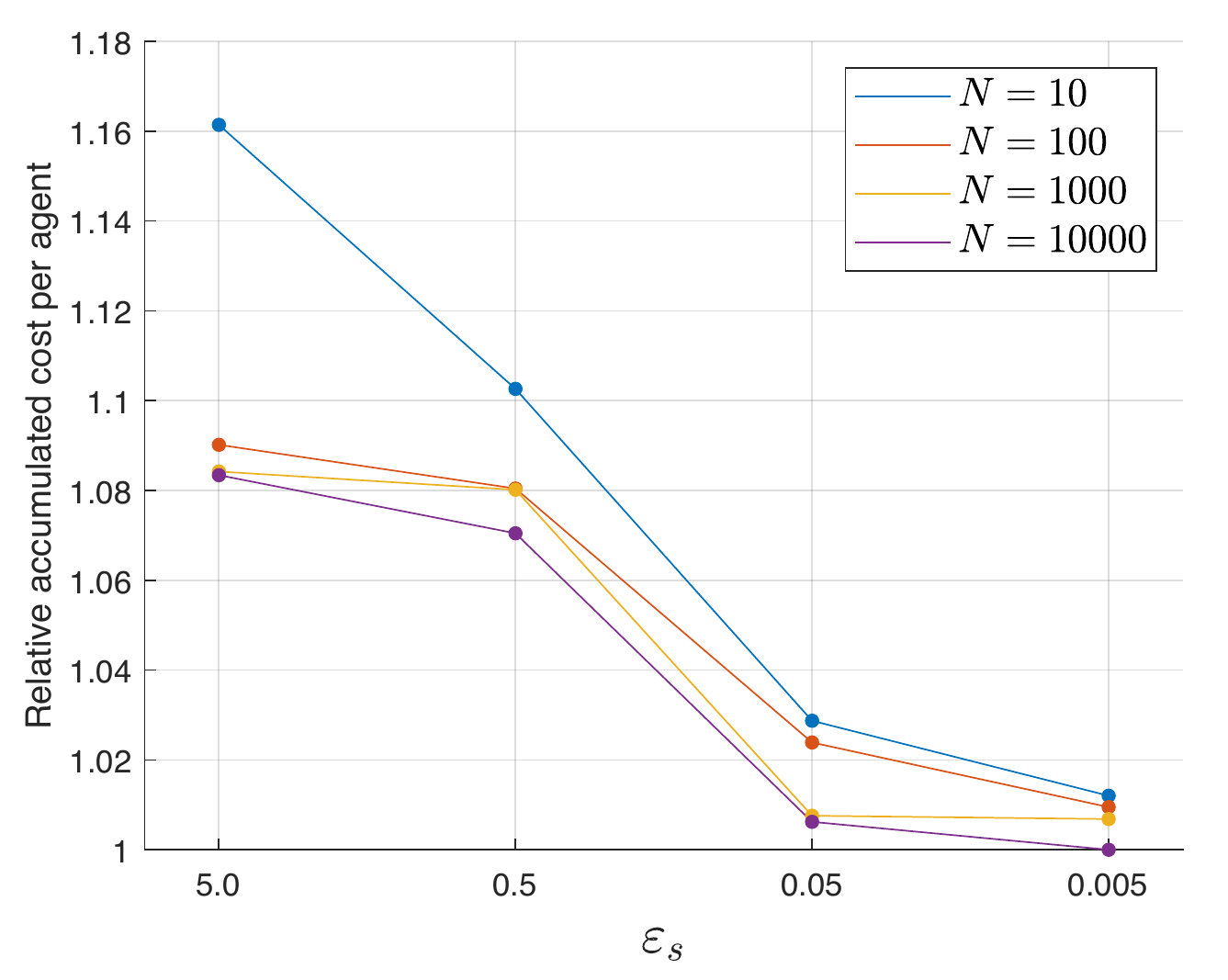}
	\caption{Relative accumulated cost per agent w.r.t. $\varepsilon_s$. Values are normalized to the lowest cost obtained ($N=10000$, $\varepsilon_s=0.005$).}
	\label{fig:avg_cost}
\end{figure}

\section{Concluding Remarks and Future Directions}
\label{sec:Conc}

We have developed a policy iteration algorithm for approximating equilibria in infinite-horizon LQ-MFGs with discounted cost. The main challenge in the algorithm development arises from the fact that the optimal control is computed backward in time. By investigating properties of the mean-field update operator (which we term the $\tau$-latent property), we can represent the mean-field trajectory at any given iteration by a finite set of parameters, resulting in a forward-in-time construction of the optimal control. The algorithm is provably convergent, with numerical results demonstrating the nature of convergence. The optimality of the computed equilibrium has been empirically studied; naturally, the optimality of the approximate equilibrium improves as the iteration index increases and the stopping threshold decreases. The results derived in this paper provide an algorithmic viewpoint of the nature of mean-field equilibria for LQ-MFGs. We believe that such insights will be useful for developing model-free RL  algorithms. Future work includes an extension to the multivariate case as well as consideration of a nonlinear/non-quadratic model (see \cite{saldi2018markov}).



\appendix
\label{sec:Append}

\begin{proof}[Proof of Lemma \ref{lem:LQT}]
Substituting $K_t = \gamma^t p_t$ and $g_t = \gamma^t \lambda_t$ into (21)--(26) of \cite{yazdani2018technical} yields (similar derivations in \cite{bertsekas1995dynamic} and on p. 234 of \cite{basar1999dynamic}),
\begin{align*}
& u_t = -\gamma b (c_u + b^2 \gamma p_{t+1})^{-1} ( p_{t+1} a z_t + \lambda_{t+1}), \\
& \lambda_t = \gamma a \left(1 - \gamma p_{t+1} b^2 (c_u + \gamma b^2 p_{t+1})^{-1} \right) \lambda_{t+1} - c_z \bar{z}_t, \\
& p_t = \gamma a^2 p_{t+1} + c_z - \gamma^2 a^2 b^2 p^2_{t+1} (c_u + \gamma b^2 p_{t+1})^{-1}.
\end{align*}
Since it is an infinite horizon problem, the Riccati equation will have a steady state solution. This can be written as,
\begin{align*}
u_t &= -\gamma b (a p  z_t + \lambda_{t+1})/(c_u +  \gamma b^2 p),\\
p &= \gamma a^2 p + c_z - \gamma^2 a^2 b^2 p^2/(c_u +  \gamma b^2 p),\\
\lambda_t &= \gamma a \left(1 - \gamma p b^2/(c_u +  \gamma b^2 p) \right) \lambda_{t+1} - c_z \bar{z}_t.
\end{align*}
Defining $g_p:= -\gamma b /(c_u + \gamma b^2 p)$ and $h_p:= a (1 + bpg_p)$, the above expressions for $u_t$ and $\lambda_t$ correspond to \eqref{eq:u_t} and \eqref{eq:lambda_t}, respectively. Rearranging and grouping $p$ terms in the above expression yields \eqref{eq:s}, with unique positive solution \eqref{eq:closedform}.
\end{proof}

\begin{proof}[Proof of Lemma \ref{lem:bdd_and_hurwitz}]
	1) First, we show that $\gamma \lvert h_p \rvert < 1$. It is well known 
	\cite{bertsekas1995dynamic} that the DARE for variables $(a_c, b_c, s_c, 
	r_c)$ 
	and average cost is $\hat{k} = a^2_c \hat{k} + s_c - a^2_c \hat{k}^2 b^2_c 
	/(r_c 
	+ b^2_c \hat{k})$. If $b_c \neq 0$ and $s_c > 0$, then this equation will 
	have a 
	positive solution. Moreover, the optimal feedback gain is $l_c := -a_c b_c 
	\hat{k}/(r_c + b_c^2 \hat{k})$ and the closed-loop gain $\lvert  a_c + b_c 
	l_c 
	\rvert < 1$. By using a change of variables with $a_c = \sqrt{\gamma} a, b_c 
	= 
	b, s_c = c_z, r_c = c_u/\gamma$, the equation \eqref{eq:s} is recovered with 
	$
	\hat{k} = p$. Hence there exists a unique positive solution for \eqref{eq:s}, given by \eqref{eq:closedform}. 
	Moreover 
	$\lvert  a_c + b_c l_c \rvert = \left\lvert \sqrt{\gamma}a - \frac{ 
	\gamma\sqrt{\gamma}a b^2 p}{c_u + \gamma b^2 p} \right\rvert =\sqrt{\gamma} 
	\left\lvert a(1 + b p g) \right\rvert= \sqrt{\gamma} \lvert h_p \rvert < 1$ 
	and 
	thus $\gamma \lvert h_p\rvert < \sqrt{\gamma}|h_p| < 1$. From 
	\eqref{eq:lambda_t}, recursing backwards yields $\lambda_0 = \left(\gamma 
	h_p\right)^t \lambda_t - c_z\sum_{s=0}^{t-1} \left( \gamma h_p \right)^s  
	\bar{z}_s$. Under the assumption $\lambda_0  = -c_z \sum_{s=0}^{\infty} 
	(\gamma 
	h_p)^s\bar{z}_{s}$, it follows that $\left(\gamma h_p\right)^t \lambda_t = -
	c_z 
	\sum_{s = t}^{\infty} \left( \gamma h_p \right)^s \bar{z}_s \Rightarrow 
	\lambda_t = -c_z \sum_{s = 0}^{\infty} \left( \gamma h_p \right)^s \bar{z}
	_{t+s}
	$. As $\bar{z} \in \ell^\infty$ there exists some $0 \leq \bar{z}_{\infty} < 
	\infty$ s.t. $\lvert \bar{z}_t \rvert \leq \bar{z}_{\infty}$. This 
	translates to 
	$\lvert \lambda_t \rvert \leq c_z \bar{z}_{\infty} / (1-\gamma h_p)$ for all 
	$t$. Hence $\lambda \in \ell^{\infty}$. \newline
	2) The closed-loop dynamics of $z_t$ under the cost-minimizing control are, 
	$z_{t+1} = a z_t + b g_p (p a z_t + \lambda_{t+1}) + \omega_t$. Using this 
	equation recursively, the expression for $z_t$ in terms of $z_0$ is, $z_t = 
	h_p^{t-1} z_0 + bg_p \sum^{t-1}_{s=0} h_p^s \lambda_{t-s} + 
	\sum^{t-1}_{s=0} 
	h_p^s\omega_{t-1-s}$. The expression for $\EE \big[ \big(z_t - \bar{z}_t 
	\big)^2 \big]$ is thus,
	\begin{align*}
	&\EE \big[ \left(z_t - \bar{z}_t \right)^2 \big] = \big( h_p^{t-1} \big)^2 
	({\sigma_0^2} + \nu^2_0) + \sigma_w^2 \sum^{t-1}_{s=0} h_p^s + \\
	& \bigg( bg_p \sum^{t-1}_{s=0} h_p^s \lambda_{t-s} - \bar{z}_t \bigg)^2 + 2 h_p^{t-1} \nu_0 \bigg( bg_p \sum^{t-1}_{s=0} h_p^s \lambda_{t-s} - \bar{z}_t \bigg).
	\end{align*}
	Assumption \ref{asm:boundforcontract} implies that $\lvert h_p \rvert < 1$. 
	Furthermore, since $\bar{z}\in\ell^{\infty}$ and $\lambda\in\ell^{\infty}$, 
	there exist constants $\bar{z}_{\infty},\lambda_{\infty}<\infty$ such that $\bar z_t\le\bar{z}_{\infty}$ and 
	$\lambda_t\le\lambda_{\infty}$ for all $t$. Thus,
	\begin{align*}
	&\EE \left[ \left(z_t - \bar{z}_t \right)^2 \right] \leq \left( h_p^{t-1} \right)^2 (\sigma^2_0 + \nu^2_0) + \sigma_w^2 t \\
	& + \left( bg_p t \lambda_{\infty} - \bar{z}_{\infty} \right)^2 + 2 h_p^{t-1} \nu_0 \left( bg_p t \lambda_{\infty} - \bar{z}_{\infty} \right).
	\end{align*}
	Similarly, $\EE \left[ u^2_t \right]$ is bounded above as,
	\begin{align*}
	&\EE \big[ u^2_t \big] = (a g_p p)^2 \big[ \big( h_p^{t-1} \big)^2 
	\big(\sigma^2_0 + \nu^2_0 \big) + \sum^{t-1}_{s=0} \sigma_w^2 h_p^{2s} \\
	& + bg_p h_p^s \lambda_{t-s} + 2 h_p^{t-1} \nu_0 bg_p h_p^s \lambda_{t-s} 
	\big] + g_p^2 \lambda^2_{t+1} \\
	& + 2a g_p^2 p \lambda_{t+1} \big( h_p^{t-1} \nu_0 + bg_p \sum^{t-1}_{s=0} 
	h_p^s \lambda_{t-1} \big) \\
	& \leq (a g_p p)^2 \left[ \left( h_p^{t-1} \right)^2 \left(\sigma^2_0 + \nu^2_0 \right) + \sigma_w^2 t \right. \\
	& \left. + \left( bg_p t \lambda_{\infty} \right) + 2 h_p^{t-1} \nu_0 bg_p t \lambda_{\infty} \right] + g_p^2 \left(\lambda_{\infty}\right)^2 \\
	& + 2a g_p^2 p \lambda_{\infty} \left( h_p^{t-1} \nu_0 + bg_p t \lambda_{\infty} \right).
	\end{align*}
	Since the optimal cost is, $\sum_{t=0}^{\infty} c_z \EE \big[ \gamma^t \big( 
	z_t - \bar{z}_t \big)^2 \big] + c_u \EE \big[ \gamma^t u^2_t \big]$, and $
	\sum_{t=0}^{\infty} t \gamma^t = \frac{\gamma}{(1-\gamma)^2},\,\sum_{t=0}
	^{\infty} \gamma^t \big( h_p^{t-1} \big)^2 = \frac{1}{h_p^2 \big( 1- \gamma 
	h_p^2 \big)},\,\sum_{t=0}^{\infty} t \gamma^t h_p^{t-1} = \frac{\gamma}{( 1 
	- 
	\gamma h_p )^2}$ it can be concluded that the optimal cost is bounded.
\end{proof}

\begin{proof}[Proof of Lemma \ref{lem:contract}]
	Let us define two mean-fields $\bar{z}, \hat{z} \in \ell^{\infty}$ and their next iterates $\bar{z}' = \mathcal{T} (\bar{z}), \hat{z}' = \mathcal{T} (\hat{z})$. Let us define the difference sequences $\delta_t = \bar{z}_t - \hat{z}_t$ and $\delta'_t = \bar{z}'_t - \hat{z}'_t$. Using \eqref{eq:MF-update}, the equation expressing the connection between $\delta_t$ and $\delta'_t$ is 
	$\delta'_{t+1} = h_p  \delta_t - c_zb g_p   \sum_{s=0}^{\infty} \left( \gamma h_p  \right)^{s} \delta_{t+1+s}$.
	Hence,
	\begin{align*}
	 \lVert \delta' \rVert_{\infty}& \leq 
	\lVert \delta \rVert_{\infty} \Big( \big\lvert h_p \big\rvert + \Big\lvert 
	c_z b g_p  \sum_{s=0}^{\infty} \big( \gamma h_p \big)^{s} \Big\rvert \Big)\\
	& \leq \lVert \delta \rVert_{\infty} \big( \big\lvert h_p \big\rvert + 
	\big\lvert c_zb g_p/(1 - \gamma h_p) \big\rvert \big) = \lVert \delta 
	\rVert_{\infty}T_p
	\end{align*}
	where the last inequality follows from $\gamma|h_p|<1$ (see Lemma \ref{lem:bdd_and_hurwitz}). By Assumption \ref{asm:boundforcontract}, $\mathcal{T}$ is a contraction.
\end{proof}

\begin{proof}[Proof of Theorem \ref{thm:MFE_T_equivalence}]
	Consider an MFE $( \mu^*, \bar{z}^*)$ that satisfies Definition 
	\ref{def:mfe}. 	Then, by definition, $\mu^* = \Phi(\bar{z}^*)$. The second 
	part of Definition \ref{def:mfe} states 
	that $\bar{z}^* = \Lambda(\mu^*)$. Thus $\bar{z}^* =\Lambda(\Phi(\bar{z}^*)) 
	= \mathcal{T}(\bar{z}^*)$. Now let us prove the converse.    Consider a 
	mean-field 
	$\bar{z}^*$ which is the fixed point of $\mathcal{T}$ i.e. $\bar{z}^* = 
	\mathcal{T}(\bar{z}^*)$. Then if $\mu^*$ 
	is the cost-minimizing control for $\bar{z}^*$ i.e. $\mu^* = \Phi(\bar{z}^*)
	$, $(\mu^*, \bar{z}^*)$ is an MFE since (1) $\mu^* = \Phi(\bar{z}^*)$, and 
	(2) $\Lambda(\mu^*) = \Lambda(\Phi(\bar{z}^*)) = \mathcal{T}(\bar{z}^*) = 
	\bar{z}^*$.
\end{proof}

\begin{proof}[Proof of Lemma \ref{lem:LTI2LTI}] For $t \in \{\tau,\tau+1,\ldots\}$ using \eqref{eq:MF-update} and the fact that $|r|\le 1$, we can write $x'_{t+1} = h_p x_t - c_zb g_p  \sum_{s=0}^{\infty} \left( \gamma h_p r 
	\right)^{s} r x_{t} = \hat r_px_{t} \label{eq:z_k+1_t+1}$ where $\hat r_p :=  h_p - \frac{c_zb g_p r}{1 - \gamma h_p r}$. Similarly, $x'_{t+2}$ is generated as $x'_{t+2} = \hat r_p x_{t+1} = \hat r_p r x_t$ 	for all $t \in \{\tau,\tau+1,\ldots\}$. Grouping terms, we obtain $x'_{t+2} = r x'_{t+1}$ for all $t \in \{\tau,\tau+1,\ldots\}$. 
\end{proof}

\begin{proof} [Proof of Theorem \ref{thm:approx_mfe_bound}]
	We first state and prove in Lemma \ref{lem:inf_norm_diff} below that the 
	expression in the stopping condition of the algorithm $\max_{0 \leq t \leq 
	k} 
	\big\lvert \bar{z}^{(k)}_t - \bar{z}^{(k-1)}_t \big\rvert$ is {equal} to {$
	\big\lVert\bar{z}^{(k)} - \bar{z}^{(k-1)} \big\rVert_{\infty}$}. This is due 
	to 
	the fact that $\bar{z}^{(k)}$ and $\bar{z}^{(k-1)}$ both follow stable 
	linear 
	dynamics for $ t \geq k$.
	\begin{lemma} \label{lem:inf_norm_diff} $\displaystyle\max_{0 \leq t \leq k} 
	\big\lvert \bar{z}^{(k)}_t - \bar{z}^{(k-1)}_t \big\rvert = 
	\big\lVert\bar{z}
	^{(k)} - \bar{z}^{(k-1)} \big\rVert_{\infty}$.
	\end{lemma}
	\emph{Proof.}
		By definition $\big\lVert\bar{z}^{(k)} - \bar{z}^{(k-1)} 
		\big\rVert_{\infty} = \sup_{t \geq 0} \big\lvert \bar{z}^{(k)}_t - 
		\bar{z}
		^{(k-1)}_t \big\rvert$. 
		Hence for all $t \geq k$, $\big\lvert \bar{z}^{(k)}_t - \bar{z}^{(k-1)}_t \big\rvert = \big\lvert r^{t-k} \big(\bar{z}^{(k)}_{k} - \bar{z}^{(k-1)}_{k}\big) \big\rvert\leq \big\lvert\bar{z}^{(k)}_{k} - \bar{z}^{(k-1)}_{k}\big\rvert$. Using this property, $\big\lVert\bar{z}^{(k+1)} - \bar{z}^{(k)} \big\rVert_{\infty} = \sup_{t \geq 0} \big\lvert \bar{z}^{(k+1)}_t - \bar{z}^{(k)}_t \big\rvert = \max_{0 \leq t \leq k+1} \big\lvert \bar{z}
		^{(k+1)}_t - \bar{z}^{(k)}_t \big\rvert.\hspace{4em}\hfill\square$
		
	Since $\Ts$ is contractive with a fixed point of $\bar z^*$,
	\begin{align} \label{eq:contract_property}
	\big\lVert \bar{z}^{(k+1)} - \bar{z}^* \big\rVert_{\infty} \leq T_p 
	\big\lVert \bar{z}^{(k)} - \bar{z}^* \big\rVert_{\infty}
	\end{align}
for any $k=1,2,\ldots$. 
The algorithm terminates at iteration $k$ when $\big\lVert \bar{z}^{(k+1)} - 
\bar{z}^{(k)} \big\rVert_{\infty}  <  \varepsilon_s (1-T_p)/T_p$. 
Thus,
	\begin{align*}
	\varepsilon_s (1 - T_p)/T_p & > \big\lVert \bar{z}^{(k+1)} - \bar{z}
	^{(k)} \big\rVert_{\infty}\\ 
	& \ge \big\lVert \bar{z}^{(k)} - \bar{z}^* \big\rVert_{\infty} - \big\lVert 
	\bar{z}^{(k+1)} - \bar{z}^* \big\rVert_{\infty} \\
	& \ge \frac{1}{T_p}\big\lVert \bar{z}^{(k+1)} - \bar{z}^* 
	\big\rVert_{\infty} - \big\lVert \bar{z}^{(k+1)} - \bar{z}^* 
	\big\rVert_{\infty} \\
	& = \big(1 - T_p \big) \big\lVert \bar{z}^{(k+1)} - \bar{z}^* 
	\big\rVert_{\infty}/T_p.
	\end{align*}
	Hence, $\big\lVert \bar{z}^{(k+1)} - \bar{z}^* \big\rVert_{\infty} < 
	\varepsilon_s$ for any $\varepsilon_s > 0$. Now we prove the bound on the 
	number 
	of iterations. If the number of iterations is $k > K(\varepsilon_s)$, then, 
	\begin{align} \label{eq:step_to_calc_k_bound}
	&k > \big(\log \varepsilon_s - \log\big\lVert\bar{z}^{(0)} - \bar{z}^{*} 
	\big\rVert_{\infty}\big)/\log T_p \nonumber \\
	&\Leftrightarrow k \log T_p < \log \varepsilon_s - \log\big\lVert\bar{z}
	^{(0)} - \bar{z}^{*} \big\rVert_{\infty}  \nonumber \\
	&\Leftrightarrow T_p^k < \frac{\varepsilon_s}{\big\lVert\bar{z}^{(0)} - 
	\bar{z}^{*} \big\rVert_{\infty}} \Leftrightarrow T_p^k \big\lVert\bar{z}
	^{(0)} - 
	\bar{z}^{*} \big\rVert_{\infty} < \varepsilon_s.
	\end{align}
	The inequality flip in the second step is due to the fact that $T_p < 1$ 
	(Assumption \ref{asm:boundforcontract}) and $\log T_p < 0$. From  
	\eqref{eq:contract_property} $\big\lVert\bar{z}^{(k)} - \bar{z}^{*} 
	\big\rVert_{\infty} \leq T_p^k \big\lVert\bar{z}^{(0)} - \bar{z}^{*} 
	\big\rVert_{\infty}$ and using the inequality 
	\eqref{eq:step_to_calc_k_bound}, $
	\big\lVert\bar{z}^{(k)} - \bar{z}^{*} \big\rVert_{\infty} < \varepsilon_s$.%
\end{proof}

\bibliographystyle{IEEEtran} 
\bibliography{references,MARL_Springer_1,MARL_Springer_2,RL}

\end{document}